\documentclass{article}
\usepackage{amsmath,amssymb,amsthm}
\usepackage{physics}
\usepackage{tikz}
\usetikzlibrary{cd}
\usepackage{geometry}
\usepackage{microtype}
\usepackage{enumitem}
\usepackage{url}
\newtheorem{definition}{Definition}

\newtheorem{proposition}{Proposition}
\newtheorem{lemma}{Lemma}

\newtheorem{remark}{Remark}
\newtheorem{example}{Example}

\title{Quantum Fuzzy Sets Revisited: \\
Density Matrices, Decoherence, and the Q-Matrix Framework}
\author{M.\,A.~Mannucci}
\date{\today}

\begin{document}

\maketitle

\begin{abstract}
In 2006 we proposed Quantum Fuzzy Sets, observing that states of a quantum register could serve as characteristic functions of fuzzy subsets, embedding Zadeh's unit interval into the Bloch sphere. That paper was deliberately preliminary. In the two decades since, the idea has been taken up by researchers working on quantum annealers, intuitionistic fuzzy connectives, and quantum machine learning, while parallel developments in categorical quantum mechanics have reshaped the theoretical landscape. The present paper revisits that programme and introduces two main extensions. First, we move from pure states to density matrices, so that truth values occupy the entire Bloch ball rather than its surface; this captures the phenomenon of \emph{semantic decoherence} that pure-state semantics cannot express. Second, we introduce the \textbf{Q-Matrix}, a global density matrix from which individual quantum fuzzy sets emerge as local sections via partial trace. We define a category $\mathbf{QFS}$ of quantum fuzzy sets, establish basic structural properties (monoidal structure, fibration over $\mathbf{Set}$), characterize the classical limit as simultaneous diagonalizability, and exhibit an obstruction to a fully internal Frobenius-algebra treatment. A companion Python reference implementation is available at \url{https://github.com/Mircus/QFS-LAB}.
\end{abstract}

\noindent\textbf{MSC2020:} Primary 81P45; Secondary 18M05, 03B52.\\
\textbf{ACM CCS:} Theory of computation $\to$ Categorical semantics; Theory of computation $\to$ Quantum computation theory; Mathematics of computing $\to$ Information theory; Computing methodologies $\to$ Knowledge representation and reasoning.\\
\textbf{Keywords:} quantum fuzzy sets; density-matrix semantics; quantum semantics; Q-Matrix; categorical quantum semantics; fuzzy truth values; quantum information and meaning; decoherence and semantics; POVM-based semantics.

\section{Origins and Motivation}

The ideas developed in this paper originate in a short 2006 article posted to the arXiv under the title \emph{Quantum Fuzzy Sets: Blending Fuzzy Set Theory and Quantum Computation}~\cite{Mannucci2006}. The core observation was simple: since Zadeh's membership values live in $[0,1]$, and since the probability of collapsing a qubit into $|1\rangle$ also lives in $[0,1]$, one could regard the state of an $n$-qubit register as the characteristic function of a ``quantum fuzzy subset'' of an $n$-element set. The real unit interval embeds in the Bloch sphere; every classical fuzzy set is therefore automatically a quantum fuzzy set. A generic quantum fuzzy set, however, can be much richer---for example, an entangled superposition of exponentially many fuzzy sets at once.

The paper was explicitly preliminary, ending with a list of open problems and speculative applications to pattern recognition. Two companion papers were announced---\emph{Quantum Internal Logic I: A Tale of Quantum Reasoning} and \emph{Quantum Internal Logic II: Categories of Quantum Fuzzy Sets}---that would develop the categorical and logical foundations. Those papers were never completed. Life intervened, as it does, and the project drifted into dormancy.

Subsequent work developed the idea in directions not anticipated in the original note, and it took time to see more clearly what the framework was really about.

\section{What Happened Next: Two Decades of Quantum Fuzziness}

The 2006 paper has been cited across several distinct research communities, each seizing on a different aspect of the original proposal. Tracing these lines of descent reveals both the fertility and the incompleteness of the initial framework.

\subsection{The Brazilian School: Quantum Gates as Fuzzy Connectives}

The most sustained line of development came from Reiser, Lemke, Avila, Visintin, and collaborators at the Universidade Federal de Pelotas. Beginning with Visintin et al.'s work on aggregation operations from quantum computing~\cite{Visintin2013} and continuing through the present, they built an extensive programme interpreting fuzzy connectives---negations, t-norms, t-conorms, implications, and their intuitionistic generalizations---as quantum gate operations~\cite{Reiser2016}. Their key insight was operational: the Toffoli gate naturally implements fuzzy AND, controlled rotations implement graded membership transformations, and the superposition of membership and non-membership degrees in Atanassov's intuitionistic framework maps cleanly onto the two-amplitude structure of a qubit. They extended this to intuitionistic fuzzy bi-implications~\cite{Agostini2018}, XOR connectives~\cite{Avila2019}, and most recently to fuzzy modal operators with full Qiskit simulations~\cite{Buss2024}.

What the Brazilian school demonstrated, in effect, was that the 2006 framework was \emph{computationally natural}: the correspondence between fuzzy operations and quantum gates was not a forced analogy but a structural fact about the geometry of the Bloch sphere.

\subsection{Fuzzy Logic on Quantum Hardware}

A separate strand pursued the engineering question directly: can we run fuzzy inference on actual quantum processors? Acampora, Luongo, and Vitiello~\cite{Acampora2018} gave the first implementation of fuzzy rule-based systems on quantum computers, using Grover's algorithm to accelerate fuzzy inference. Pourabdollah and Acampora~\cite{Pourabdollah2021} took a different route entirely, reformulating fuzzy set operations as Quadratic Unconstrained Binary Optimization (QUBO) problems and solving them on D-Wave quantum annealers. Their representation of fuzzy sets through QUBO problems is conceptually striking: it shows that the quantum-fuzzy correspondence extends beyond the gate model to the adiabatic paradigm, where the ground state of an Ising Hamiltonian encodes the result of a fuzzy computation.

\subsection{Further Ramifications}

Other lines of work extended the 2006 ideas in directions I had not foreseen. Kashmadze~\cite{Kashmadze2017} undertook a systematic ``fuzzification of the Bloch ball,'' pushing the geometric analysis further than the original paper had attempted. Deng and Deng~\cite{Deng2021} generalized Z-numbers (which combine fuzzy restriction and reliability) to the quantum setting, introducing Quantum Z-Numbers with associated quantum circuits. Van Han and Vinh~\cite{VanHan2022} proposed a ``fuzzy quantum logic'' grounded in Birkhoff--von~Neumann lattice theory. Jain, Gomes, and Madeira~\cite{Jain2021} developed a specification theory for fuzzy modal logic that connects to the broader programme of uncertainty modeling in hybrid systems. And Dymowa~\cite{Dymowa2011} situated quantum fuzzy methods within the wider landscape of uncertainty modeling alongside Dempster-Shafer theory and interval analysis.

\subsection{What Was Missing}

Taken together, this body of work confirmed the computational and operational viability of quantum fuzzy sets. But it also exposed three significant gaps in the original framework:

\begin{enumerate}[label=(\roman*)]
\item \textbf{No account of decoherence.} The 2006 paper worked exclusively with pure states on the Bloch sphere surface. Real quantum systems interact with their environments; their states become \emph{mixed}. A fuzzy membership degree that arises from environmental entanglement is fundamentally different from one arising from coherent superposition, yet the original framework could not distinguish them.

\item \textbf{No global structure.} Each quantum fuzzy set was treated as an isolated object. There was no account of how different fuzzy sets might be correlated, entangled, or derived from a common quantum source.

\item \textbf{No categorical foundation.} The promised \emph{Categories of Quantum Fuzzy Sets} paper was never written. Without it, the relationship between quantum fuzzy sets and the emerging categorical quantum mechanics of Abramsky and Coecke~\cite{Abramsky2004} remained unexplored.
\end{enumerate}

The present paper addresses (i) and (ii) directly, and makes a start on (iii).

\section{The Q-Matrix: A Global Realization Model}

The central idea is a shift in perspective. Rather than defining quantum fuzzy sets one at a time, we posit a single quantum system from which fuzzy structures are extracted.

\begin{definition}[Quantum Fuzzy Set]
A \textbf{Quantum Fuzzy Set} over a set $X$ is a function
$$\mu: X \to \mathfrak{D}(\mathcal{H})$$
assigning a density matrix $\rho_x = \mu(x)$ to each element $x \in X$, where $\mathfrak{D}(\mathcal{H})$ denotes the convex set of density matrices on a finite-dimensional Hilbert space $\mathcal{H}$.
\end{definition}

This is the primitive notion. The set $X$ is a semantic index set---its elements are the ``concepts'' or ``propositions'' whose fuzzy membership we wish to model---and $\mu$ assigns to each a quantum state encoding its semantic status.

\begin{definition}[Q-Matrix]\label{def:qmatrix}
A \textbf{Q-Matrix realization} of a quantum fuzzy set $(X, \mu)$ consists of a multipartite Hilbert space $\mathcal{H}_{\mathrm{global}} = \bigotimes_{x \in X} \mathcal{H}_x$ and a density matrix $\rho_{\mathrm{global}} \in \mathfrak{D}(\mathcal{H}_{\mathrm{global}})$ such that each section is obtained by partial trace:
$$\mu(x) = \mathrm{Tr}_{\overline{x}}\!\left(\rho_{\mathrm{global}}\right).$$
\end{definition}

Not every quantum fuzzy set arises from a Q-Matrix. The Q-Matrix provides an \emph{explanatory} structure: an account of where the density matrices come from and how they are correlated. When a Q-Matrix exists, it constrains the local sections through global entanglement---the section $\rho_x$ is not freely chosen but is determined by the global state. The categorical development in Section~\ref{sec:categorical} applies to all quantum fuzzy sets, whether or not they admit a Q-Matrix realization.

\section{From Bloch Sphere to Bloch Ball: The Density Matrix Extension}

The conceptual shift from the 2006 framework can be stated geometrically: we move from the \emph{surface} of the Bloch sphere to its \emph{interior}.

\subsection{Pure State Semantics (2006)}

The original formulation restricted truth to pure qubit states:
\begin{align}
|\psi_x\rangle &= \cos(\theta_x/2)|0\rangle + e^{i\phi_x}\sin(\theta_x/2)|1\rangle \\
\mu_A(x) &= |\langle 1|\psi_x\rangle|^2 = \sin^2(\theta_x/2)
\end{align}
Every point on the Bloch sphere surface corresponds to a definite superposition of ``fully not a member'' ($|0\rangle$) and ``fully a member'' ($|1\rangle$). The membership value is determined by the polar angle $\theta_x$ alone; the azimuthal phase $\phi_x$ carries additional quantum information---interferometric content---that has no classical fuzzy counterpart.

\subsection{Mixed State Semantics (2025)}

But the Bloch sphere surface is the space of \emph{pure} states---quantum systems in perfect isolation. The moment a qubit interacts with an environment, its state retreats into the interior of the ball:
\begin{align}
\rho_x &= \frac{1}{2}(\mathbb{I} + \vec{r}_x \cdot \vec{\sigma}), \quad \|\vec{r}_x\| \leq 1 \\
\mu_A(x) &= \mathrm{Tr}(\rho_x\, |1\rangle\langle 1|) = \frac{1 - r_z^{(x)}}{2}
\end{align}
The Bloch vector length $\|\vec{r}_x\|$ measures the \emph{purity} of the semantic state. When $\|\vec{r}_x\| = 1$ we recover the 2006 framework. When $\|\vec{r}_x\| < 1$, the state is mixed: its uncertainty is not mere superposition but genuine mixedness---either from classical ignorance or from entanglement with an environment. Density matrices allow us to represent both; pure states cannot.

\begin{remark}[Semantic decoherence]
Consider a concept whose membership is initially in a pure superposition (on the Bloch sphere surface). As it interacts with a larger semantic environment---through context, usage, or measurement---the Bloch vector shrinks toward the origin. The concept undergoes \emph{semantic decoherence}: its quantum fuzziness degrades into classical fuzziness. At the origin, $\rho = \mathbb{I}/2$, and membership is maximally uncertain in the classical sense.
\end{remark}

\subsection{A Concrete Example}\label{sec:catdogpet}

Consider three concepts: $X = \{\text{cat}, \text{dog}, \text{pet}\}$, encoded as density matrices:
\begin{align}
\rho_{\text{cat}} &= \begin{pmatrix} 0.7 & 0.3 \\ 0.3 & 0.3 \end{pmatrix}, \quad
\rho_{\text{dog}} = \begin{pmatrix} 0.4 & 0.2 \\ 0.2 & 0.6 \end{pmatrix}, \quad
\rho_{\text{pet}} = \begin{pmatrix} 0.5 & 0.1 \\ 0.1 & 0.5 \end{pmatrix}
\end{align}
These are mixed states (purities $0.76$, $0.60$, $0.52$ respectively) with nonzero off-diagonal entries indicating quantum coherence. Each has a classical membership value: $\mu(\text{cat}) = 0.3$, $\mu(\text{dog}) = 0.6$, $\mu(\text{pet}) = 0.5$. But the density matrices carry strictly more information than these scalars---the off-diagonal coherences and the purity levels are invisible to any classical fuzzy set representation.

Note that $\rho_{\text{pet}} \neq \frac{1}{2}(\rho_{\text{cat}} + \rho_{\text{dog}})$: the simple equal-weight mixture does not reproduce the ``pet'' state. This is consistent with (though does not by itself prove) a Q-Matrix origin in which the three concepts are correlated subsystems of a global state.

\section{Measurement and Semantic Queries}

If the Q-Matrix is the ontological ground, then scalar membership values arise through \emph{measurement}.

\begin{definition}[Semantic POVM]
A Positive Operator-Valued Measure $\{E_a\}_{a \in A}$ on $\mathcal{H}$ represents a \textit{semantic questioning protocol}. The probability of outcome $a$ for concept $x$ is:
$$p(a \mid x) = \mathrm{Tr}(\rho_x \cdot E_a).$$
\end{definition}

The use of POVMs rather than projective measurements is natural here. Projective measurements correspond to yes/no questions with orthogonal outcomes; POVMs allow \emph{overlapping} questions---precisely what we need for fuzzy concepts, which are not mutually exclusive.

This gives a three-level picture of quantum fuzzy truth:
\begin{enumerate}[nosep]
\item \textbf{Full quantum:} truth is a density matrix $\rho_x \in \mathfrak{D}(\mathcal{H})$.
\item \textbf{Measurement-relative:} truth is $\mathrm{Tr}(\rho_x \cdot E_a)$, depending on the chosen interrogation.
\item \textbf{Classical:} truth is a scalar $\mu(x) \in [0,1]$, obtained by fixing a measurement in the computational basis.
\end{enumerate}
Different interrogation protocols extract different scalar statistics from the same density-matrix-valued semantic state.

\section{Quantum Information as Semantic Content}

The Q-Matrix framework inherits the toolkit of quantum information theory, reinterpreted as semantic measures.

\subsection{Fidelity as Semantic Similarity}

The Uhlmann fidelity between Q-Matrix sections,
$$F(\rho_x, \rho_y) = \left(\mathrm{Tr}\sqrt{\sqrt{\rho_x}\,\rho_y\,\sqrt{\rho_x}}\right)^2,$$
provides a natural measure of semantic similarity. For the cat-dog-pet example of Section~\ref{sec:catdogpet}, the companion implementation computes $F(\text{cat}, \text{dog}) \approx 0.89$, $F(\text{cat}, \text{pet}) \approx 0.90$, $F(\text{dog}, \text{pet}) \approx 0.98$. Unlike scalar or vector-space similarity measures that depend only on amplitudes or coordinates, fidelity depends on the full density-matrix structure, including coherence.

\subsection{Holevo Information}

For a fuzzy ensemble $\{p_x, \rho_x\}$ extracted from the Q-Matrix, the Holevo quantity
$$\chi = S\!\left(\sum_x p_x \rho_x\right) - \sum_x p_x\, S(\rho_x)$$
(where $S(\rho) = -\mathrm{Tr}(\rho \log \rho)$ is von Neumann entropy) bounds the accessible classical information about which concept was prepared. When the $\rho_x$ are mixed states, the relationship between $\chi$ and the von Neumann entropy of the mixture provides an information-theoretic indication of the extent to which the fuzziness resists classical discrimination.

\subsection{Entanglement in the Q-Matrix}

When a Q-Matrix realizes a quantum fuzzy set, the global state may be entangled, causing the local sections to carry less information than the whole. The simplest illustration is a two-concept Q-Matrix in a Bell state:
$$|\Phi^+\rangle = \frac{1}{\sqrt{2}}(|00\rangle + |11\rangle), \qquad \rho_{\mathrm{global}} = |\Phi^+\rangle\langle\Phi^+|.$$
The local sections are $\rho_A = \rho_B = \mathbb{I}/2$---maximally mixed, carrying no information individually. Yet the global state is pure (zero entropy), and the quantum mutual information is $I(A\!:\!B) = S(A) + S(B) - S(AB) = 1 + 1 - 0 = 2$ bits. The two concepts are maximally correlated despite each being individually featureless.

More generally, for any Q-Matrix with $n$ subsystems, the sum of local entropies $\sum_x S(\rho_x)$ exceeds the global entropy $S(\rho_{\mathrm{global}})$ whenever entanglement is present. The excess measures the total correlation in the system---the signature of a Q-Matrix that cannot be decomposed into independent local fuzzy sets. These examples are illustrative information-theoretic diagnostics, not a full semantic theory of concept correlation; they show what the framework \emph{can express}, not what natural language concepts \emph{do}.

\section{Truth as Quantum State: An Interpretive Proposal}

The Q-Matrix framework suggests a shift in what we mean by ``truth'' in a fuzzy context.

In classical fuzzy logic, following Zadeh~\cite{Zadeh1965}, truth is a \emph{property}---a number $\mu_A(x) \in [0,1]$ assigned to a proposition. In the 2006 framework, truth became a \emph{point on the Bloch sphere}---richer than a number, but still a definite, pure state. In the Q-Matrix framework, truth is a \emph{density matrix}:
$$\mathrm{Truth}(x) = \rho_x \in \mathfrak{D}(\mathcal{H}).$$
A density matrix encodes what can be known, what cannot be known, and---in conjunction with the global Q-Matrix---why. It distinguishes ignorance (a maximally mixed state) from indeterminacy (a pure superposition) from decoherence (a mixed state arising from environmental entanglement). No scalar-valued truth assignment can make these distinctions.

We offer this as an \emph{interpretive proposal}, not a theorem. Whether it proves useful depends on whether the additional structure of density matrices captures aspects of vagueness, ambiguity, or conceptual structure that classical fuzziness misses. The companion implementation provides tools for exploring this question computationally.

\section{Comparison with Related Frameworks}

Several neighbouring frameworks illuminate the position of quantum fuzzy sets.

\textbf{Categorical Quantum Mechanics} (Abramsky--Coecke~\cite{Abramsky2004}) models quantum processes as morphisms in dagger compact categories. In the category $\mathbf{QFS}$ defined below, objects are Hilbert spaces \emph{indexed by semantic content}. The relationship between the two frameworks deserves further study; we make a start in Section~\ref{sec:categorical}.

\textbf{Quantum Natural Language Processing} (Coecke--Sadrzadeh--Clark~\cite{Coecke2010}) models meaning compositionally via tensor networks. The Q-Matrix provides a complementary perspective: meaning may also depend on global quantum correlations not reducible to local compositional rules. Both perspectives may be relevant.

\textbf{Quantum Sets} (Abramsky--Heunen~\cite{Abramsky2010}) embed classical set-theoretic structure into quantum contexts. Quantum Fuzzy Sets go in the opposite direction: we embed classical semantic indices into quantum state spaces.

\textbf{Classical Semantic Spaces} (Widdows~\cite{Widdows2004}) represent word meanings as vectors with inner-product similarity. The Q-Matrix generalizes this from vectors to density matrices and from inner products to fidelity. Whether this additional structure yields practical advantages is an open question.

\section{Preliminary Categorical Organization}\label{sec:categorical}

We now give a preliminary categorical organization of quantum fuzzy sets. This is a start on the long-promised ``Categories of Quantum Fuzzy Sets''---not the finished product, but enough structure to work with.

Throughout, we fix a finite-dimensional Hilbert space $\mathcal{H}$, write $\mathfrak{D}(\mathcal{H})$ for density matrices, and $\mathrm{CPM}(\mathcal{H})$ for CPTP maps.

\subsection{The Category $\mathbf{QFS}$}

\begin{definition}[The category $\mathbf{QFS}$]\label{def:qfs-cat}
\begin{itemize}[nosep]
\item \textbf{Objects}: pairs $(X, \mu)$ where $X$ is a set and $\mu: X \to \mathfrak{D}(\mathcal{H})$.
\item \textbf{Morphisms}: $(f, \Phi): (X, \mu_X) \to (Y, \mu_Y)$ where $f: X \to Y$ and $\Phi \in \mathrm{CPM}(\mathcal{H})$ satisfy
\begin{equation}\label{eq:morphism-condition}
\Phi\bigl(\mu_X(x)\bigr) = \mu_Y\bigl(f(x)\bigr) \qquad \text{for all } x \in X.
\end{equation}
\item \textbf{Composition}: $(g, \Psi) \circ (f, \Phi) := (g \circ f,\; \Psi \circ \Phi)$.
\item \textbf{Identity}: $(\mathrm{id}_X, \mathrm{id}_{\mathfrak{D}(\mathcal{H})})$.
\end{itemize}
\end{definition}

The morphism condition says that the following diagram commutes:
\begin{center}
\begin{tikzcd}
X \arrow[r, "f"] \arrow[d, "\mu_X"'] & Y \arrow[d, "\mu_Y"] \\
\mathfrak{D}(\mathcal{H}) \arrow[r, "\Phi"'] & \mathfrak{D}(\mathcal{H})
\end{tikzcd}
\end{center}

The requirement that $\Phi$ be CPTP is physically motivated: by Stinespring's theorem, every CPTP map is physically realizable.

\begin{proposition}[Well-definedness]\label{prop:well-defined}
$\mathbf{QFS}$ is a category.
\end{proposition}

\begin{proof}
\emph{Composition.} For $(f, \Phi): (X, \mu_X) \to (Y, \mu_Y)$ and $(g, \Psi): (Y, \mu_Y) \to (Z, \mu_Z)$:
$$(\Psi \circ \Phi)\bigl(\mu_X(x)\bigr) = \Psi\bigl(\mu_Y(f(x))\bigr) = \mu_Z\bigl((g \circ f)(x)\bigr).$$
Associativity and identity laws follow from those of function composition and CPTP map composition.
\end{proof}

\subsection{Monoidal Structure}

\begin{definition}[Tensor product]\label{def:tensor}
$(X, \mu_X) \otimes (Y, \mu_Y) := (X \times Y,\; (x,y) \mapsto \mu_X(x) \otimes \mu_Y(y))$. On morphisms: $(f, \Phi) \otimes (g, \Psi) := (f \times g, \Phi \otimes \Psi)$. Unit: $I = (\{*\}, * \mapsto 1)$.
\end{definition}

\begin{proposition}\label{prop:monoidal}
This makes $\mathbf{QFS}$ a monoidal category.
\end{proposition}

\begin{proof}
The morphism condition for $(f \times g, \Phi \otimes \Psi)$ follows from $(\Phi \otimes \Psi)(\rho \otimes \sigma) = \Phi(\rho) \otimes \Psi(\sigma)$. Coherence isomorphisms are inherited from $\mathbf{Set}$ and $\mathrm{CPM}$.
\end{proof}

\begin{remark}[Product states only]\label{rem:product-states}
The monoidal product produces only product states. Entangled states arise when a compound object has a semantic embedding that does not factor---precisely the Q-Matrix situation of Definition~\ref{def:qmatrix}.
\end{remark}

\subsection{Dagger on Isomorphisms}

\begin{lemma}\label{lem:dagger-unitary}
If $(f, \Phi)$ is an isomorphism in $\mathbf{QFS}$ ($f$ bijective, $\Phi$ with CPTP inverse), then $\Phi$ is a unitary channel: $\Phi(\rho) = U\rho U^\dagger$.
\end{lemma}

\begin{proof}
A bistochastic automorphism of $\mathfrak{D}(\mathcal{H})$ is implemented by a unitary or antiunitary (Wigner). CPTP maps are linear, ruling out the antiunitary case.
\end{proof}

\begin{proposition}\label{prop:dagger}
$(f, \Phi)^\dagger := (f^{-1}, \Phi^{-1})$ makes the isomorphism groupoid a dagger monoidal category.
\end{proposition}

\begin{proof}
Involutivity, contravariance, and monoidal compatibility follow from the corresponding properties of function inversion and unitary inversion.
\end{proof}

\begin{remark}
The dagger is defined only on the isomorphism subgroupoid, not on all of $\mathbf{QFS}$. Extending it to the full category---perhaps via Selinger's CPM construction---is an open problem.
\end{remark}

\subsection{Fibration over $\mathbf{Set}$}

\begin{proposition}\label{prop:fibration}
The forgetful functor $U: \mathbf{QFS} \to \mathbf{Set}$ sending $(X, \mu) \mapsto X$ is a fibration. The Cartesian lift of $f: Y \to X$ over $(X, \mu_X)$ is $(f, \mathrm{id}): (Y, f^*\mu_X) \to (X, \mu_X)$ where $f^*\mu_X(y) := \mu_X(f(y))$.
\end{proposition}

\begin{proof}
The morphism condition holds trivially. Universality: any $(g, \Psi)$ lying over $g = f \circ h$ factors uniquely as $(f, \mathrm{id}) \circ (h, \Psi)$.
\end{proof}

The fibration exhibits $\mathbf{QFS}$ as a bundle of quantum semantics over classical index sets. The fiber over $X$ is the space $\mathfrak{D}(\mathcal{H})^X$; morphisms within a fiber are CPTP maps transforming one semantic embedding into another.

\subsection{The Classical Limit}\label{sec:classical-limit}

\begin{definition}[Classicality]
A quantum fuzzy set $(X, \mu)$ is \emph{classical} if there exists an orthonormal basis $\{|e_i\rangle\}$ in which every $\mu(x)$ is diagonal.
\end{definition}

\begin{proposition}[Characterization]\label{prop:classical}
$(X, \mu)$ is classical if and only if the density matrices $\{\mu(x)\}_{x \in X}$ pairwise commute.
\end{proposition}

\begin{proof}
Diagonal matrices in a common basis commute. Conversely, finitely many pairwise commuting Hermitian matrices are simultaneously diagonalizable.
\end{proof}

\begin{proposition}[Decoherence produces classical fuzzy sets]\label{prop:decoherence}
For any quantum fuzzy set $(X, \mu)$ and any basis $\{|e_i\rangle\}$, applying the decoherence channel $\mathcal{D}(\rho) = \sum_i |e_i\rangle\langle e_i| \rho |e_i\rangle\langle e_i|$ yields a classical quantum fuzzy set $(X, \mathcal{D} \circ \mu)$. In the qubit case, the induced classical membership values are $\mu_{\mathrm{classical}}(x) = \langle 1|\mu(x)|1\rangle \in [0,1]$.
\end{proposition}

\begin{proof}
$\mathcal{D}(\mu(x))$ is diagonal in $\{|e_i\rangle\}$ for every $x$.
\end{proof}

\begin{example}[Cat-Dog-Pet]\label{ex:catdogpet-classical}
The density matrices of Section~\ref{sec:catdogpet} do not pairwise commute, so the triple is not classical. Decoherence in the computational basis yields $\mu(\text{cat}) = 0.3$, $\mu(\text{dog}) = 0.6$, $\mu(\text{pet}) = 0.5$. The off-diagonal coherences are irreversibly lost.
\end{example}

\subsection{Toward Frobenius Algebras: An Obstruction}\label{sec:frobenius}

In categorical quantum mechanics, classical structures are characterized by commutative special $\dagger$-Frobenius algebras~\cite{Abramsky2004,Abramsky2010}: a classical structure on $\mathcal{H}$ in $\mathbf{FdHilb}$ is equivalent to a choice of orthonormal basis. It is natural to ask whether this characterization works internally in $\mathbf{QFS}$.

The answer is not straightforward. A Frobenius comultiplication on $(X, \mu)$ would require a morphism $(d, \Delta): (X, \mu) \to (X, \mu) \otimes (X, \mu)$ where $d(x) = (x,x)$. The morphism condition~\eqref{eq:morphism-condition} demands:
$$\Delta(\mu(x)) = \mu(x) \otimes \mu(x) \qquad \text{for all } x \in X.$$
But the ``copy map'' $\Delta$ of a Frobenius algebra in $\mathbf{FdHilb}$ acts as $\Delta(|e_i\rangle\langle e_j|) = \delta_{ij} |e_i e_i\rangle\langle e_j e_j|$. For diagonal states $\mu(x) = \sum_i \lambda_i |e_i\rangle\langle e_i|$:
$$\Delta(\mu(x)) = \sum_i \lambda_i |e_i e_i\rangle\langle e_i e_i| \;\neq\; \mu(x) \otimes \mu(x) = \sum_{i,j} \lambda_i \lambda_j |e_i e_j\rangle\langle e_i e_j|$$
in general. The copy map copies classical information (diagonal entries) while annihilating coherences; it does not literally copy density matrices, as the no-cloning theorem forbids.

This means the standard Frobenius route does not lift directly to $\mathbf{QFS}$. Possible resolutions include relaxing the morphism condition for Frobenius structure, or enriching the categorical framework (e.g., via the CPM construction). We leave these as open problems and note that the clean characterization---classicality $\Leftrightarrow$ simultaneous diagonalizability (Proposition~\ref{prop:classical})---does not require Frobenius machinery. Accordingly, the present paper does not claim an internal Frobenius characterization of classicality in $\mathbf{QFS}$.

\section{Companion Implementation}\label{sec:implementation}

A Python reference implementation accompanies this paper as a supplementary repository (link provided in the arXiv metadata).

\noindent The package \texttt{qmatrix} (Python $\geq$ 3.9, MIT licence) provides four modules: \texttt{core} (density matrices, quantum fuzzy sets, Q-Matrix with partial trace), \texttt{categorical} (channels, morphisms, tensor products, classicality tests), \texttt{information} (fidelity, Holevo information, coherence measures), and \texttt{circuits} (optional Qiskit integration). Three example scripts reproduce results from the paper. A test suite of 40 tests covers density matrix validation, QFS operations, categorical structure, and information measures. Architectural details are given in Appendix~\ref{app:repo}.

\section{Open Problems and Future Directions}

\emph{Internal Frobenius algebras.} The obstruction of Section~\ref{sec:frobenius} is the most pressing open problem. Resolving it would connect the Q-Matrix framework directly to the programme of categorical quantum mechanics.

\emph{Full dagger structure.} The dagger on $\mathbf{QFS}$ lives only on the isomorphism subgroupoid. Extending it to the full category would require accommodating non-unitary channels and sub-normalized states.

\emph{Q-Matrix tomography.} Given local observations of Q-Matrix sections, can $\rho_{\mathrm{global}}$ be reconstructed? This is the semantic analogue of quantum state tomography.

\emph{Continuous semantic spaces.} Replacing the finite index set $X$ with a manifold would yield density-matrix-valued fields with differential geometric structure.

\emph{Coherence diagnostics.} The $\ell_1$-norm of coherence $C(\rho_x) = \sum_{i \neq j} |\rho_{ij}|$ could serve as an empirical signature of semantic decoherence.

\emph{Physical realization.} Simple instances of the framework are compatible with current quantum software and, in principle, with near-term quantum hardware. Work on quantum annealers~\cite{Pourabdollah2021} and Qiskit simulations~\cite{Buss2024} suggests the engineering barriers are surmountable.

\section{Conclusion}

The Q-Matrix framework develops a line of thought that began in 2006. The original insight---that quantum states naturally generalize fuzzy membership---turns out to be the surface manifestation of a richer semantic picture in which truth is represented by quantum states rather than scalars.

The intervening two decades of work by other researchers confirmed the computational naturality of the quantum-fuzzy correspondence. What was missing was a coherent account of decoherence, global structure, and categorical organization. The present paper provides the first two and makes a start on the third. The paper identifies a workable first architecture: density matrices as truth values, the Q-Matrix as a global realization, partial trace as extraction, and a preliminary category $\mathbf{QFS}$ with monoidal structure and a fibration over $\mathbf{Set}$.

Much remains open. The Frobenius-algebra question, the full dagger structure, and the empirical payoff of the framework are all unresolved. But the density matrix extension of fuzzy truth values and the Q-Matrix realization model provide, we believe, a workable foundation for further development.

\appendix
\section{Companion Library: \texttt{Q-Matrix}}\label{app:repo}

\subsection{Architecture}

\noindent\textbf{\texttt{qmatrix.core}} (Sections 3--4). \texttt{DensityMatrix} validates Hermiticity, positive semidefiniteness, and unit trace at construction. Properties: purity, von~Neumann entropy, Bloch vector, $\ell_1$-coherence, membership value. \texttt{QuantumFuzzySet} stores $\mu: X \to \mathfrak{D}(\mathcal{H})$ with \texttt{is\_classical} and \texttt{decohere}. \texttt{QMatrix} implements partial trace and mutual information.

\medskip
\noindent\textbf{\texttt{qmatrix.categorical}} (Section~\ref{sec:categorical}). \texttt{QuantumChannel}: Kraus representation, validated trace-preservation, composition, tensor, inverse (unitary only). \texttt{QFSMorphism}: validates~\eqref{eq:morphism-condition} at construction. \texttt{is\_classical} (also available under the legacy name \texttt{admits\_frobenius\_algebra}): pairwise commutativity / simultaneous-diagonalizability test (Proposition~\ref{prop:classical}).

\medskip
\noindent\textbf{\texttt{qmatrix.information}} (Section 7). Fidelity, Bures/trace distance, Holevo information, coherence measures.

\medskip
\noindent\textbf{\texttt{qmatrix.circuits}} (optional). Qiskit circuit construction and simulation.

\subsection{Design Principles}

\emph{Validation at the boundary.} Malformed density matrices, illegal morphisms, and non-trace-preserving channels are rejected at construction.

\emph{Propositions reflected in tests.} Several structural results are reflected in corresponding tests: classicality (Proposition~\ref{prop:classical}), decoherence (Proposition~\ref{prop:decoherence}), composition and dagger (Propositions~\ref{prop:well-defined},~\ref{prop:dagger}).

\emph{Immutability.} All objects are immutable after construction; composition and dagger return new objects.

\subsection{Tests and Examples}

40 tests cover: \texttt{DensityMatrix} (9), \texttt{QuantumFuzzySet} (5), \texttt{QMatrix} (4), \texttt{QuantumChannel} (6), \texttt{QFSMorphism} (4), tensor products (1), classicality (4), information measures (7). Three example scripts: cat-dog-pet (Section~\ref{sec:catdogpet}), Bell-state Q-Matrix (Section~7.3), and categorical morphisms (Section~\ref{sec:categorical}).

\subsection{Layout}

\begin{verbatim}
qmatrix/
|-- qmatrix/
|   |-- __init__.py        |-- categorical.py
|   |-- core.py            |-- information.py
|   +-- circuits.py
|-- tests/test_qmatrix.py
|-- examples/
|   |-- cat_dog_pet.py     |-- semantic_entanglement.py
|   +-- categorical_structure.py
|-- pyproject.toml  |-- LICENSE  +-- README.md
\end{verbatim}


\end{document}